\documentclass[conference]{IEEEtran}
\usepackage[utf8]{inputenc}
\usepackage{enumitem}
\usepackage{amsthm,amssymb,amsmath}
\usepackage{graphicx,multirow}
\usepackage{hanging,colortbl}
\usepackage{listings,xcolor,sgame,float,bbm}
\usepackage{siunitx}
\usepackage{comment}
\usepackage{color,soul,mathtools}
\usepackage[noend]{algpseudocode}
\usepackage{algorithm}
\usepackage{stackengine}
\usepackage{siunitx}
\usepackage{booktabs}

\newtheorem{problem}{\textbf{Problem}}

\newtheorem{theorem}{\textbf{Theorem}}

\newcounter{lemtheorem}
\DeclareMathOperator*{\argmin}{\text{arg min}\,}

\makeatletter
\newcommand{\ostar}{\mathbin{\mathpalette\make@circled\star}}
\newcommand{\make@circled}[2]{%
  \ooalign{$\m@th#1\smallbigcirc{#1}$\cr\hidewidth$\m@th#1#2$\hidewidth\cr}%
}
\newcommand{\smallbigcirc}[1]{%
  \vcenter{\hbox{\scalebox{0.77778}{$\m@th#1\bigcirc$}}}%
}
\makeatother

\algrenewcommand\algorithmicprocedure{}

\title
{Iterative Recommendations based on Monte Carlo Sampling and Trust Estimation in Multi-Stage Vehicular Traffic Routing Games} 

\author{\IEEEauthorblockN{Doris E. M. Brown, Venkata Sriram Siddhardh Nadendla, and Sajal K. Das}
\vspace{0.5ex}
\IEEEauthorblockA{
Department of Computer Science 
\\ 
Missouri University of Science and Technology 
\\ 
Rolla, MO 65409, USA 
\\ Email: \{deby3f, nadendla, sdas\}@mst.edu}
\vspace{-4ex}
}

\begin{document}

\maketitle

\begin{abstract}
The shortest-time route recommendations offered by modern navigation systems fuel selfish routing in urban vehicular traffic networks and are therefore one of the main reasons for the growth of congestion. In contrast, intelligent transportation systems (ITS) prefer to steer driver-vehicle systems (DVS) toward system-optimal route recommendations, which are primarily designed to mitigate network congestion. However, due to the misalignment in motives, drivers exhibit a lack of trust in the ITS. This paper models the interaction between a DVS and an ITS as a novel, multi-stage routing game where the DVS exhibits dynamics in its trust towards the recommendations of ITS based on counterfactual and observed game outcomes. Specifically, DVS and ITS are respectively modeled as a travel-time minimizer and network congestion minimizer, each having nonidentical prior beliefs about the network state. A novel approximate algorithm to compute the Bayesian Nash equilibrium, called \emph{ROSTER} (Recommendation Outcome Sampling with Trust Estimation and Re-evaluation), is proposed based on Monte Carlo sampling with trust belief updating to determine the best response route recommendations of the ITS at each stage of the game. Simulation results demonstrate that the trust prediction error in the proposed algorithm converges to zero with a growing number of multi-stage DVS-ITS interactions and is effectively able to both mitigate congestion and reduce driver travel times when compared to alternative route recommendation strategies.
\end{abstract}


\section{Introduction}
The continual escalation of global traffic congestion leads to the pressing need for efficient traffic routing approaches \cite{schrank20152015,pishue20212021}. The ever-growing urbanization rate \cite{seidel2020rush} and inefficiency in selfish routing \cite{roughgarden2001stackelberg} are two major reasons behind the rapid increase in traffic congestion, particularly in the last decade. This paper tries to tackle the problem of congestion reduction by mitigating network inefficiencies due to selfish routing, where agents seek to maximize their individual utilities without any consideration for optimality in terms of any social welfare metric. Initially, navigation systems were thought to be a solution to mitigate traffic congestion, as users can find and choose low-congestion routes. However, most navigation systems (e.g. Google Maps and Waze) typically recommend shortest travel-time routes \cite{mehta2019google,laor2022waze}, which further boost selfish routing as drivers minimize individual travel times without considering their impact on overall road network congestion. 

With the rise of connected and autonomous vehicles (CAVs), intelligent transportation systems (ITS) have been proposed to develop advanced driver assist systems (ADAS), which provide additional safety and driving support for vehicles. If an ITS infrastructure is equipped with computational power, in addition to just connected sensors, the ITS is adept at presenting route recommendations that are strategically designed to mitigate network congestion \cite{cheng2020mitigating}. Despite this potential of ITS to recommend {\em welfare-optimizing routes} that mitigate a network metric such as congestion, carbon emissions, or safety, modern vehicles are not fully autonomous and still require a human driver, primarily to operate the vehicle, who may not be fully compliant with ITS recommendations due to distrust. Therefore, it is natural to consider an ITS that iteratively recommends a driver-vehicle system (DVS) to steer it toward welfare-optimizing routes.

An ITS 
typically employs route recommendations to influence the behavior of DVS. Although drivers use GPS-based navigation systems to receive route recommendations, many of them have reported not trusting recommendations from such systems \cite{yamsaengsung2018towards}, \cite{amin2018evaluating}, \cite{trapsilawati2019human}. User distrust in navigation systems originates from the perceived unreliability of applications to provide route recommendations that meet user expectations in terms of travel time or traversal path \cite{yamsaengsung2018towards}. Although modern navigation systems aim to leverage users' selfish tendencies, yet they still generate distrust among a significant portion of users. On the other hand, if they aim to optimize social welfare, potentially increasing individual travel times, the ITS 
would likely face more significant levels of DVS distrust due to this misalignment in motives. 
Therefore, ITS route recommendation algorithms must consider DVS trust to recommend routes that are appealing to both the ITS goal of mitigating network congestion and the DVS goal of minimizing travel time.

{\em Routing games}, in which selfish drivers are routed through a network along congestion-aware paths, are typically modeled as Stackelberg games, where a routing strategy is found through an approximate algorithm to determine the best response route recommendation of the ITS \cite{krichene2017stackelberg}, \cite{roughgarden2001stackelberg}, \cite{bonifaci2010stackelberg}. These approaches, while successful in reducing network congestion, typically fail to consider a spectrum of trust at the DVS. Hence, they are unable to propose alternative route recommendations if the DVS rejects a recommendation, thereby offering no intermediate solution between the initial recommendation of the ITS and the selfish route of the DVS. On the other hand, {\em multi-stage games} offer a promising approach to modeling this interaction, since self-interested ITS and DVS can cooperate with each other to some extent to reach a solution that is beneficial for both players. This motivates our work.

\subsection{Contributions} 
In this paper, the interaction between ITS and DVS, hereafter referred to as the system and vehicle, respectively, is modeled as a multi-stage game where the system and driver are network congestion minimizer and travel-time minimizer, respectively, and the driver has dynamic trust unknown to the system. 
To the best of our knowledge, this is the first work in this direction.
Our novel contributions are as follows.
\begin{enumerate}
    \item We model a \textbf{\emph{multi-stage Stackelberg game}} between a driver and a system in which their motives are misaligned and the driver exhibits trust dynamics towards the system. 
    \item We develop an \textbf{\emph{approximate algorithm}}, called ROSTER (Recommendation Outcome Sampling with Trust Estimation and Re-evalutation), to approximate the Bayesian Nash equilibrium. It is based on Monte Carlo sampling with trust belief updating to determine the best response route recommendations at the system at each game stage.
    \item We present \textbf{\emph{theoretical results}} on the worst-case and best-case costs incurred by the ROSTER-based system.
    \item We validate the ROSTER algorithm with \textbf{\emph{diverse, realistic simulation experiments}} using traffic network data from Manhattan and Sioux Falls (USA) and observed that ROSTER performs better than alternative route recommendation and route selection strategies in mitigating both network congestion and the driver's travel time. 
\end{enumerate}

The paper is organized as follows. Section \ref{Sect: Related Work} reviews the related work. Section \ref{Sec: Problem Formulation} formulates the system-driver interaction as a routing game and introduces the Bayesian Nash equilibrium. Section \ref{Sec: Proposed Methodology} presents the ROSTER algorithm and best response approximations. Section \ref{Sect: Experimental Evaluation} covers performance metrics and simulation results. Finally, Section \ref{Sect: Conclusion} summarizes the work and suggests future research directions.

\section{Related Work} 
\label{Sect: Related Work}
Several existing approaches have been proposed to mitigate congestion in traffic networks. They can broadly be categorized under marginal cost pricing, information revelation, and Stackelberg routing.

In {\bf marginal cost pricing} approaches, drivers are charged a toll or tax based on the marginal additional cost they impose on the network travel time. While marginal cost pricing has long been known to mitigate the impact of selfish routing on social welfare in traffic networks \cite{vickrey1963pricing}, such approaches typically assume that the drivers' routing decisions are static which the system knows in order to deploy accurate tolling. To combat these assumptions, efforts have been made to develop marginal cost routing techniques considering drivers' dynamic routing decisions \cite{sharon2017real}. While such approaches are theoretically sound, drivers and some governments express an aversion to marginal cost pricing \cite{harrington2001overcoming}, making it infeasible to implement in practice in democratic societies.

The {\bf information revelation} approaches reveals the network state information to the drivers, publicly or privately. Although some of these approaches have shown that congestion is mitigated under certain conditions \cite{acemoglu2018informational}, other studies have noted that making additional information available to drivers can lead to increased congestion under certain conditions \cite{wu2017informational}, \cite{wu2021value}. Furthermore, providing information regarding the state of the entire network can lead to cognitive overload, leading the driver to make sub-optimal routing decisions given their inability to process all information shared with them. 

In \textbf{Stackelberg routing game} approaches, while some drivers are centrally routed strategically, others selfishly choose their routes. These games mitigate network congestion by guiding traffic toward a system-optimal solution under trust-based stochastic \cite{Brown2024TASR} or deterministic \cite{roughgarden2001stackelberg, bonifaci2010stackelberg} driver compliance. Existing approaches assume that the system knows each driver's compliance or trust probability and prior route preferences. However, in practice, a system lacks complete knowledge of driver trust because modern navigation systems provide multiple route recommendations rather than a single one that is accepted or rejected. Prior Stackelberg solutions either fail to consider driver trust when recommending routes \cite{roughgarden2001stackelberg, bonifaci2010stackelberg}, or assume driver trust is known to the system \cite{Brown2024TASR}.

\vspace{2pt}
{\em Given drivers' aversion to marginal cost pricing, potential overload from information-based routing, and existing familiarity with route recommendations, our proposed work in this paper develops a novel multi-stage Stackelberg game framework in which a traffic system strategically proposes multiple sequential route recommendations to mitigate network congestion while appealing to drivers' selfish interests.}

\section{Problem Formulation \label{Sec: Problem Formulation}}
This section introduces the traffic network scenario, multi-stage Stackelberg game that models the interaction between the intelligent transportation system (ITS) and driver-vehicle system (DVS) within the network, and
Bayesian Nash equilibrium representing the leader's optimal route recommendation given the follower's strategic response.

\subsection{Traffic Network}
Consider a transportation network as a graph $\mathcal{G} = \{\mathcal{V}, \mathcal{E}\}$, where the vertex-set $\mathcal{V}$ represents a set of physical locations in the network, and the edge-set $\mathcal{E} = \{e_1, \cdots, e_K\}$ represents a set of $K$ roads between pairs of locations in $\mathcal{V}$. For each origin-destination pair, denoted by $(o,d)$, a finite set $\mathcal{R} = \{r_1, \cdots, r_N\}$ of $N$ edge-disjoint simple paths (routes) exists between the origin and the destination. Let a given route $r_i \in \mathcal{R}$ consist of a set of edges $\mathcal{E}_{r_i} \in \mathcal{E}$. Then the travel time of $r_i$ is computed as:
\begin{equation*}
\label{Eqn: Route Travel Time}
T(r_i) = \sum_{e_j \in \mathcal{E}_{r_i}} t(e_j, f_{e_j}),
\end{equation*}
and the network congestion is computed as: $\sum_{r_i \in \mathcal{R}} T(r_i)$,
where $f_{e_j}$ denotes the flow of traffic on an edge $e_j$, and $t(e_j, f_{e_j})$ represents the travel time of an edge $e_j$ with the flow $f_{e_j}$. Now
$t(e_j, f_{e_j})$ 
is computed as follows using the Bureau of Public Roads (BPR) \cite{manual1964bureau} function:
\begin{equation*}
\label{Eqn: Travel Time BPR}
t_{e_j}(f_{e_j}) = t^{ff}_{e_j} \left [ 1 + \lambda \left (\frac{f_{e_j}}{c_{e_j}} \right)^\beta \right ],
\end{equation*}
where $t^{ff}_{e_i}$ denotes the free-flow travel time of $e_i$; $c_{e_i}$ is the capacity of $e_i$; and $\lambda$ and $\beta$ are constants commonly assumed to be $0.15$ and $4$, respectively. The state of the network is defined as $\boldsymbol{\pi} = \{\pi(f_{e_j})\}_{e_j \in \mathcal{E}}$, which represents the true probability distribution of traffic volume on each edge. 

\subsection{Players and Game Progression}
Assume a driver aims to traverse the network via a shortest-travel-time $(o,d)$-route, and 
the driver has a prior belief $\boldsymbol{q} = {q(f_{e_j})}_{e_j \in \mathcal{E}}$ about traffic volume along an edge $e_j$ from prior experience. Here, $q(f_{e_j}) = \mathcal{P}(f_{e_j} = f)$ denotes the driver's probabilistic belief about the flow on an edge $e_j$. 
The driver computes the expected travel time of an edge as follows:
\begin{equation}
\mathbb{E}_q(t_{e_j}) = t_{e_j}(\mathbb{E}_q(f_{e_j})),
\label{Eqn: Expected Edge TT}
\end{equation}
where the expected travel time of a route is given by
\begin{equation}
\mathbb{E}_q(T(r_i)) = \displaystyle \sum_{e_j \in \mathcal{E}_{r_i}} \mathbb{E}_q(t_{e_j}).
\label{Eqn: Expected Route TT}
\end{equation}

Assume there exists a traffic system capable of providing route recommendations to drivers with the goal of minimizing congestion within the network. Like the driver, the system also has a prior belief $\boldsymbol{p} = {p(f_{e_j})}_{e_j \in \mathcal{E}}$ about the traffic volume along each edge, with the help of sensing infrastructure that monitors traffic. The system computes the expected travel time of an edge $e_j$ as $\mathbb{E}_p(t_{e_j})$ and the expected travel time of a route $r_i$ as $\mathbb{E}_p(T(r_i))$, following Equations (\ref{Eqn: Expected Edge TT}) and (\ref{Eqn: Expected Route TT}), respectively. Then the system evaluate the congestion of the network as
\begin{equation*}
\psi_S(\mathcal{G}) = \sum_{e_j \in \mathcal{E}} t_{e_j}(f_{e_j}).
\label{Eqn: system Network Congestion}
\end{equation*}

Assume the driver starts with an initial trust score $\alpha_0 \in (0,1]$ for the system at $t = 0$. Since a driver with $\alpha_0 = 0$ would not seek recommendations, we assume $\alpha_0 \neq 0$. Assume at time step $t = 0$, the driver requests a route recommendation from the system, initiating the multi-stage game. Let $t \in \{0, \dots, t^*\}$ represent the time step and $m \in \{0, \dots, m^*\}$ denote the stage of an interaction, where an interaction consists of two time steps. A driver may engage in multiple stages of interaction before ultimately choosing a route without the system's input. To reflect this, assume a maximum of $m^* \in \{1, 2, 3\}$ stages, where $m^* = 3$ by default unless the driver accepts a recommendation earlier.

At each stage of the interaction, the system, denoted by the subscript $S$, constructs a recommendation $a_S^{t} = (r_i, \mathbb{E}_p(T(r_i)))$ which is communicated to the driver. 
Upon receiving the recommendation, the driver, denoted by the subscript $D$, evaluates its route options, and takes an action $a_D^{t + 1} \in \mathcal{R}$. The driver is assumed to \emph{accept} or \emph{reject} the recommendation if $a_D^{t+1} = r_i$ or $a_D^{t+1} \not = r_i$, respectively. Each $(a_S^t, a_D^{t+1})$ pair constitutes a stage of the interaction, and after each stage concludes, the set of available routes is assumed to be $\mathcal{R}_{m+1} = \mathcal{R} - r_i$. At any time $t$ in the interaction, history of the interaction $h^t = (\cdots, a_S^t \text{ or } a_D^t)$ is a sequence of actions taken by the system and driver.

\begin{figure}
\centering
\includegraphics[width=0.47\textwidth]{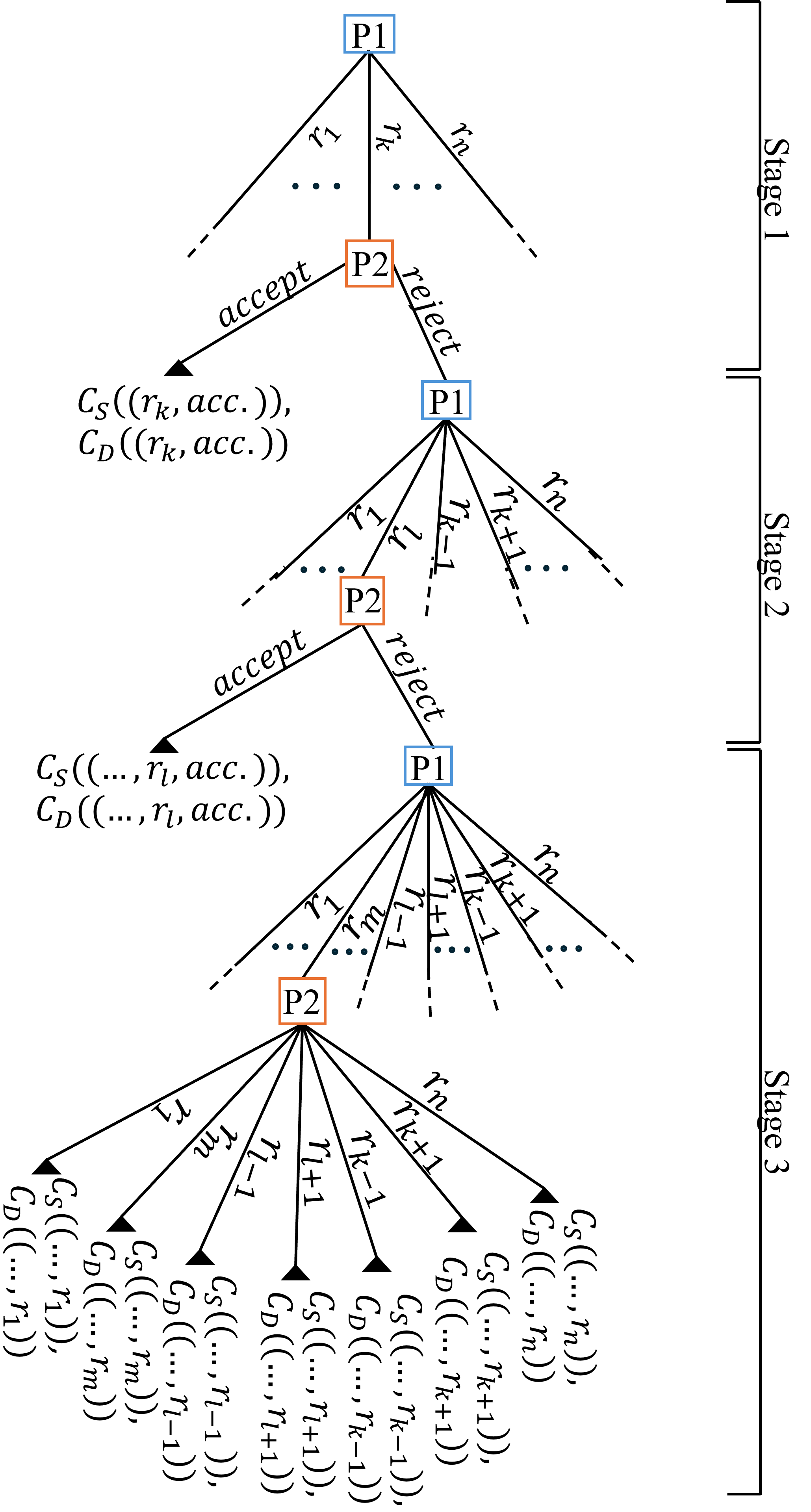}
\vspace{-0.1in}
\caption{Three-Stage Interaction between the system (P1) and driver (P2)}
\label{Fig: Negotiation Game 3 Stage}
\vspace{-0.2in}
\end{figure}

At stage $m^*$, the interaction concludes, and the driver incurs a cost of
\begin{equation}
C_D(h^{t^*}) = \gamma_D^{m^*-1} \cdot T(r_D^{t^*}),
\label{Eqn: driver Cost}
\end{equation}
where $r_D^{t = 2m}$ is the final action taken by the driver in $h^{t^*}$, 
and the constant $\gamma_D$ denotes the discomfort of the driver as more stages are needed to determine a route choice. Similarly, the system is assumed to incur a cost of 
\begin{equation}
C_S(h^{t^*}) = \gamma_S^{m^*-1} \cdot \psi_S(\mathcal{G} | r_D^{t^*}),
\label{Eqn: system Cost}
\end{equation}
where $\gamma_S$ is the system's discomfort, and  $\psi_S(\mathcal{G} | r_D^{t = 2m})$ is the network congestion when the driver traverses route $r_D^{t = 2m}$.

Figure \ref{Fig: Negotiation Game 3 Stage} illustrates the full recommendation and decision making for a three-stage interaction. At each stage, the system (player P1) sends a recommendation to the driver (player P2) who accepts or rejects. 
Assume that in subsequent stages, previously-recommended routes are unable to be recommended again. This interaction 
repeats until a recommendation has been accepted or the maximum stage has been reached. This sequence of stages constitutes a single interaction.

\subsection{Bayesian Nash Equilibrium}
Let strategies of the system $S$ and driver $D$ be denoted by 
\begin{equation*}
\sigma_S = ((a_S^1, \cdots, a_S^{{t^*}-1} | a_D^{t^*-2} \not = a_S^{{t^*-1}}))
\end{equation*} 
and ~~~~~~~
$\sigma_D = (a_D^2(a_S^1), \cdots, a_D^t(a_S^{t^*}))$


\begin{problem}
\label{Problem: Optimization Problem}
The equilibrium between the system and the driver is given by the pair $(\sigma_S^*, \sigma_D^*)$, defined as
\begin{equation*}
\begin{array}{ll}
    \sigma_S^* & = \displaystyle \argmin_{\sigma_S} C_S(h^t(\sigma_S, \sigma_D^*)),
     \\[2ex]
    \sigma_D^* & = \displaystyle \argmin_{\sigma_D} C_D(h^t(\sigma_S^*, \sigma_D)), 
\end{array}
\end{equation*}
where $h^t(\sigma_S, \sigma_D)$ denotes the history of the game until time $t$ assuming the system and driver are playing the strategies $\sigma_S$ and $\sigma_D$, respectively.
\end{problem}

Due to incomplete information at the system regarding the driver's trust, as well as practical limitations on the ability of the system to fully evaluate all routes relating to a route request, computing an optimal strategy at the outset of an interaction with the driver is infeasible. Therefore, a solution that approximates an optimal strategy $\sigma_S^*$ must be computed.

\section{Proposed Methodology \label{Sec: Proposed Methodology}}
Typically, a single-stage Stackelberg game is formulated as a bi-level optimization problem. At the upper level, the leader optimizes its objective considering the follower's best response; while at the lower level, the follower reacts to the leader's decision to form its best response. Finding a Nash equilibrium for bi-level optimization problem, and therefore a Stackelberg game, is known to be at least NP-hard \cite{sinha2017review}.
Since the system-driver interaction takes the form of a multi-stage Stackelberg game, finding a Nash equilibrium of Problem \ref{Problem: Optimization Problem} is at least NP-hard, thus necessitating an approximate solution. 

To solve Problem~\ref{Problem: Optimization Problem} the system and the driver are required to evaluate counterfactual histories across all game stages and compute counterfactual scores for each possible action. Given that the system and the driver may not have the computational power or time to compute all counterfactual game histories, a sampling-based approach offers a practical and scalable alternative for approximating equilibrium strategies in this setting. Iterative sampling-based algorithms, including those that rely on Monte Carlo sampling, are known to converge to Nash equilibrium under certain conditions \cite{papadimitriou2008computing}, and can be used to approximate a Nash equilibrium \cite{ponsen2011computing}. Hence, the approximate best responses of the system and driver are computed using a Monte Carlo sampling-based approach that evaluates a subset of counterfactual histories of the game. 

\subsection{Best Response of the Driver \label{Sect: Driver BR}}
At each interaction stage, the best response action of the driver is found using a Monte Carlo sampling approach, where the driver evaluates a subset of possible outcomes and the histories, leading to those outcomes. Here, an outcome denotes the payoff obtained by the system and driver at the end of a game, and has a corresponding history of actions taken by each player leading to that outcome. Let $\hat{\Omega}_{m}$ denote the set of all possible future outcomes if the driver chooses to reject a recommendation sent by the vehicle at stage $m$. Eliminating the outcome in which the driver accepts the recommendation at stage $m$ implies that $|\hat{\Omega}_m| = |\Omega_m| - 1$. 
The driver is assumed to sample $g_{m,D} \leq |\hat{\Omega}_{m}|$ outcomes from $\hat{\Omega}_m$ to form a set of sampled outcomes $\Bar{\Omega}_{m,D} = \{\omega_1, \cdots, \omega_{g_{m,D}}\}$. For each sampled outcome $\omega_i$, the driver uses Equation (\ref{Eqn: driver Cost}) to compute a score $C_D(\hat{h}^{\omega_i})$ given the counterfactual history leads to the outcome, $\hat{h}^{\omega_i}$. Let the sampled outcome yielding the best counterfactual cost score be given by 
\begin{equation*}
\omega_i^* = \argmin_{\omega_i \in \Bar{\Omega}_{m,D}} C_D(\hat{h}^{\omega_i}).
\label{Eqn: driver Counterfactual Best History}
\end{equation*}
Let $a_D^{t = 2m}$ be the action taken by driver at time $t$ in the outcome $\omega_i^*$ representing the driver's best counterfactual response. 

The average counterfactual cost score for rejecting the recommendation is computed as
\begin{equation*}
\label{Eqn: Driver Outcome Sampling AVerage}
\mu_{m,D} =  \frac{\sum_{\omega_i \in \Bar{\Omega}_{m,D}} C_D(\hat{h}^{\omega_i})}{g_{m,D}}.
\end{equation*} 
Considering the recommendation $(r_i, \mathbb{E}_p(T(r_i)))$ sent by the system at stage $m$, the driver computes the expected travel time of the recommended path as 
\begin{equation*}
\mathbb{E}_{q^\prime}(T(r_i)) = \alpha_m \cdot \mathbb{E}_p(T(r_i)) + (1 - \alpha_m) \cdot \mathbb{E}_{q}(T(r_i)).
\label{Eqn: driver Updated Expected TT}
\end{equation*}
We assume that the driver takes an action at time $t = 2m$ as:
\begin{equation}
\label{Eqn: Driver Decision}
a_D^t =
    \begin{cases}
      r_i & \text{if}\ \mathbb{E}_{q^\prime}(T(r_i))  \leq \mu_{m,D}  \\
    a_D^t & \text{otherwise}.
    \end{cases}
\end{equation}
Note that if $a_D^t = r_i$, the driver accepts the system's recommendation; otherwise it rejects the recommendation.

\subsection{Best Response of the System}
Compared to a single driver, a traffic system has significantly more compute power. However, in determining its best response, when the number of routes $N$ is large, 
the system may not have enough compute power to calculate all possible histories and outcomes of an interaction within a reasonable amount of time. The system may also not have perfect knowledge of $\alpha_0$, thus requiring the system to approximate $\alpha_0$ as $\hat{\alpha}_0$ and iteratively update $\hat{\alpha}_m$ at each stage. Therefore, the system is assumed to follow a Monte Carlo sampling approach to determine counterfactually-optimal route recommendations at each stage and compute $\hat{\alpha}_m$ using a no-regret-based approach.

Assume the system samples $g_{m,S} \leq |\Omega_m|$ outcomes to form a set of sampled outcomes $\Bar{\Omega}_{m,S} = \{\omega_1, \cdots, \omega_{g_{m,S}}\}$. Let $\Bar{H} = \{h^{\omega_1}, \cdots, h^{\omega_{g_{m,S}}}\}$ denote the set of histories corresponding to sampled outcomes in $\Omega_m$. For each $\omega_i \in \Bar{\Omega}_{m,S}$, the system considers its action $a_S^t$ at time $t = 2m - 1$. Let $\mathcal{R}^{\omega_i}$ denote the set of routes corresponding to actions $a_D^t$ for all $\omega_i \in \Bar{\Omega}_{m,S}$. The system then computes a counterfactual score at the driver using $\hat{\alpha}_m$ and assuming the driver samples counterfactual outcomes in which $a_S^t$ is rejected, computes a score for each of these outcomes, determines an optimal counterfactual best response $\hat{r}_D^{t+1}$ at the driver, and computes an average score $\hat{\mu}_{m,D}$ for rejecting $a_S^t$. The system computes the predicted counterfactual action of the driver at time $t + 1$ as follows: 
\begin{equation*}
\hat{a}_D^{t+1} =
\begin{cases}
    r_{a_S^t} & \text{if}\ \hat{C}_D(h^{\omega_i}|a_S^t)\leq \hat{\mu}_{m,D}
    \\[2ex]
    \hat{r}_D^{t+1} & \text{otherwise},
\end{cases}
\label{Eqn: Predicted Rec. Acceptance}
\end{equation*}
where
\begin{equation*}
\hat{C}_D(h^{\omega_i}|a_S^t) = \gamma_D^{m-1} \cdot \mathbb{E}_{\hat{q}^\prime}(\alpha_t)
\end{equation*}
and
\begin{equation*}
 \mathbb{E}_{\hat{q}^\prime}(\alpha_t) = \hat{\alpha}_m \cdot \mathbb{E}_p(a_S^t) + (1 - \hat{\alpha}_m) \cdot \mathbb{E}_q(a_S^t). 
\end{equation*}
The system computes its own counterfactual cost score, $C_S(h^{\omega_i} | \hat{a}_D^{t+1})$, for sampled outcome $\omega_i$ assuming that the driver takes action $\hat{a}_D^{t+1}$. The score is calculated similar to Equation (\ref{Eqn: system Cost}) assuming the counterfactual resulting network congestion from the driver's predicted action. 
Assume
\begin{equation*}
h^{\omega_i^*} = \argmin_{h^{\omega_i} \in \Bar{H}} C_S(h^{\omega_i} | \hat{a}_D^{t+1}).
\end{equation*}
At time $t = 2m - 1$, the system greedily takes an action $a_S^t$ at time $t$ in $h^{\omega_i^*}$, where
\begin{equation*}
a_S^t = (r_{j^\circledast}, \mathbb{E}_p(T(r_{j^\circledast}))),
\label{Eqn: system BR}
\end{equation*}
At each stage of the interaction, the system updates $\hat{\alpha}_m$, which is described in detail in Section \ref{Sect: Trust Dynamics}, and re-evaluates its approximate best-response action using Monte Carlo sampling-based approach described above. The ROSTER algorithm (\textbf{R}ecommendation \textbf{O}utcome \textbf{S}ampling with \textbf{T}rust \textbf{E}stimation and \textbf{R}e-evaluation) is described in Algorithm \ref{Alg: ROSTER}.

\textbf{An Illustrative Example:} For simplicity, consider that $\mathcal{G}$ consists of four routes, the system and driver can respectively sample three and two outcomes, $\hat{\alpha}_0 = 0.5$, $\gamma_D^0 = 1.125$, and the game is one stage. At time $m = 0$, the driver samples three histories, where the first is $(r_1, r_1)$ with outcomes in minutes of $(12, 10)$.
Assuming $\mathbb{E}_{p}(r_1) = 8$, the system computes its prediction about the driver's belief and the counterfactual score of the driver accepting $r_1$ as $\mathbb{E}_{\hat{q}^\prime}(r_1) = 0.5 \cdot 8 + 0.5 \cdot 10 = 9$ and the counterfactual score 
$\hat{C}_D(h^{\omega_1} | r_1) = 1.125 \cdot 9 = 10.125$, respectively. The driver is assumed by the system to sample rejection histories $(r_1, r_2)$ and $(r_1, r_3)$, with counterfactual driver outcomes of 11 and 9 minutes, respectively, a rejection score of $\hat{\mu}_{1,D} = 10$, and counterfactual driver best response route of $\hat{r}_D^2 = r_3$. The system assumes that the driver compares the counterfactual scores of routes $r_1$ and $r_3$, which are 10.125 and 9 respectively, and chooses $\hat{a}_D^2 = r_3$ as the best response. 
Suppose the total network congestion is $C_S(H^{\omega_1} | r_1) = 1200$ minutes if the driver takes $r_3$. If $r_1$ is recommended, the system considers that the driver will ultimately take $r_3$ for a network congestion of 1,200 minutes. This process repeats for the system's other sampled histories, and the system ultimately recommends the sampled route $r_j^\circledast$ leading to the lowest counterfactual network congestion.

\begin{algorithm}[t]
\caption{ROSTER$(R, h^m, H^{t^*})$}\label{Alg: ROSTER}
\begin{algorithmic}[1]
\State $\Omega_m = $ \Call{GetValidOutcomes}{$h^m, H^{t^*}$}
\vspace{1ex}
\State $\Bar{\Omega}_{S,m} = $ \Call{GetSamples}{$\Omega_m, g_{m,S}$}
\vspace{1ex}
\State SampledRoutes $ = [ \ ][ \ ]$
\For{each sample $\omega_i \in \Bar{\Omega}_{S,m}$}
\State $r_{a_S^{t}} = \omega_i[0]$
\State $\mathbb{E}_{\hat{q}^\prime} = $ \Call{PredictDriverUpdatedBelief}{$\hat{\alpha}_m, r_{a_S^t}$}
\State $\hat{C}_D(h^{\omega_i}|r_{a_S^t}) = (\gamma_D)^{m-1} \cdot \mathbb{E}_{\hat{q}^\prime}(\alpha_t)$
\State $\hat{\mu}_{m,D} = $ \Call{GetDriverRejectScore}{$r_{a_S^t}$}
\State $\hat{r}_D^{t+1} = $ \Call{GetDriverRejectBR}{$r_{a_S^t}$}
\State $\hat{a}_D^{t+1} = $ \Call{GetDriverAction}{$\hat{C}_D(h^{\omega_i}|r_{a_S^t}), \hat{r}_D^{t+1}$}
\State $C_S(h^{\omega_i} | \hat{a}_D^{t+1}) = $ \Call{GetITSScore}{$h^{\omega_i}, \hat{a}_D^{t+1}$}
\State SampledRoutes$[r_{a_S^t}] = C_S(h^{\omega_i} | \hat{a}_D^{t+1})$ 
\EndFor
\State $r_j^\circledast = \Call{GetBestRecommendation}{SampledRoutes}$ 
\State \textbf{return} $r_j^\circledast$
\end{algorithmic}
\end{algorithm}

\vspace{-0.1in}
\subsection{Properties of the ROSTER algorithm}

\begin{theorem}
\label{Lemma: Roster Congestion Bounds}
The network congestion resulting from the ROSTER algorithm is given by
\begin{equation*}
    \psi_S(\mathcal{G} | r_i^\circledast) \leq \psi_S(\mathcal{G} | r_i^*) \leq \sum_{e \in \mathcal{E}} t_e(f_e + 1),
\end{equation*}
where $r_i^*$ is the optimal selfish route for the driver, and $r_i^\circledast$ is the route that minimizes network congestion.
\end{theorem}

\begin{proof}
In the worst-case routing scenario, the driver rejects all recommendations and selects $r_i^*$, leading to a maximum congestion: 
\begin{equation*}
    \psi_S(\mathcal{G} | r_i^*) = \sum_{e \in \mathcal{E}_{r_i^*}} t_e(f_e + 1) + \sum_{e \notin \mathcal{E}_{r_i^*}} t_e(f_e).
\end{equation*}
In the best-case scenario, the driver follows the system's optimal recommendation $r_i^\circledast$, leading to a minimum congestion:
\begin{equation*}
    \psi_S(\mathcal{G}| r_i^\circledast) = \min_{r \in \mathcal{R}} \sum_{e \in \mathcal{E}_r} t_e(f_e + 1) + \sum_{e \notin \mathcal{E}_r} t_e(f_e).
\end{equation*}
Thus, the network congestion falls between the worst-case and best-case congestion.
\end{proof}

\subsection{Trust Dynamics \label{Sect: Trust Dynamics}}
\subsubsection{Driver Trust Updating}
Since humans are known to deviate from Bayesian updating \cite{ortoleva2022alternatives} but have been shown to follow a recency-based weighted average update model \cite{sutton2018reinforcement, strickland2024humans} when updating their trust values, the driver is assumed to update its trust in the system as follows:
\vspace{-0.1in}
\begin{equation}
\label{Eqn: Intermediate Trust Update}
\alpha_{m} = 
\begin{cases}
\alpha_{m-1} \ \text{if} \ \hat{B}^m_D(a_D^t) = 0,
\\
(1 - \eta_{D,m})^{m} \alpha_0 + \displaystyle \sum_{i = 0}^{m- 1}\eta_D(1 - \eta_{D,m})^i C_i \ \text{otherwise,}
\end{cases}
\end{equation}
where $C_i = 1$ if the driver accepts recommendation at stage $i$, and $0$ otherwise. Here, $\eta_{D,m}$ is the driver's adaptive learning rate at stage $m$, which is influenced by the change in the driver's discomfort, and is computed as follows
\vspace{-0.09in}
\begin{equation*}
\label{Eqn: Driver Discomfort}
\eta_{D,m} = \varepsilon_D \cdot \nabla B_D^m(a_D^t),
\vspace{-0.08in}
\end{equation*}
where $\varepsilon_D$ is a constant used to scale the gradient of the regret incurred by the driver. At intermediate interaction stages (i.e., $m \not = m^*$), the driver's regret is given by
\begin{equation*}
B_D^m(a_D^t) = \mathbb{E}_p(T(r_i)) - \mu_{m,D},
\label{Eqn: driver Counterfactual Regret}
\end{equation*}
where $r_i$ denotes the route recommended by the system. On the other hand, when $m = m^*$, the driver's regret is based on true observed outcomes and computed as
\begin{equation*}
B_D^m(a_D^t) = 
\begin{cases}
T(r_i) - \mu_{m,D}, & \text{if}\ a_D^t = r_i\\
\mathbb{E}_{q^\prime}(T(r_i)) - T(r_i) & \text{otherwise}.
\end{cases}
\label{Eqn: driver True Regret}
\end{equation*}

\begin{figure}[t]
\centering
\includegraphics[width=0.47\textwidth]{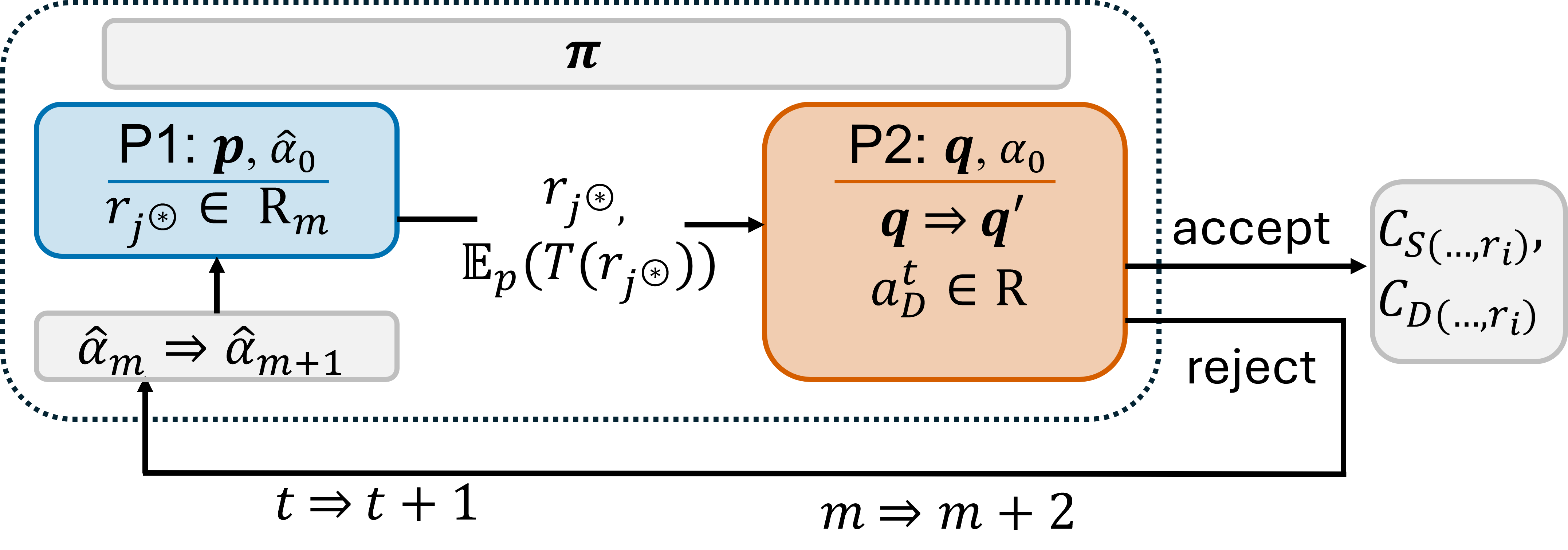}
\vspace{-0.1in}
\caption{Information exchange and decision-making within one driver-system interaction (as in Fig. \ref{Fig: Negotiation Game 3 Stage}, player P1 is the system, and player P2 the driver)}
\vspace{-0.1in}
\label{Fig:Interaction Diagram}
\vspace{-0.1in}
\end{figure}

\subsubsection{System Trust Prediction Updating}
The system is assumed to update its trust prediction using a basic regret minimization approach to update its trust prediction as
\begin{equation*}
\hat{\alpha}_m = 
\begin{cases}
    \hat{\alpha}_{m-1} + \eta_{S,m}  & \text{if}\ \hat{B}^s_D(a_D^t) < 0\\
    \hat{\alpha}_{m-1} - \eta_{S,m} & \text{if}\ \hat{B}^s_D(a_D^t) > 0 \\
    \hat{\alpha}_{m-1} + \eta_{S,0} & \text{otherwise},
\end{cases}
\label{Eqn: system Trust Prediction Update}
\end{equation*}
where $\eta_{S,0}$ is the default learning rate of the driver, and $\eta_{S,m}$ is computed as
\begin{equation*}
\eta_{S,m} = \varepsilon_S \cdot \hat{B}^s_D(a_D^t),
\label{Eqn: driver Learning Rate}
\end{equation*}
where $\hat{B}^s_D(a_D^t)$ is the driver's prediction of the system's true regret at the current stage.
It is given by
\begin{equation*}
\hat{B}^s_D(a_D^t) = 
\begin{cases}
    T(r_{a_D^t}) - \hat{\mu}_{m,D} & \text{if} \ a_D^t = r_i \\
    \hat{\mathbb{E}}_{q^\prime}(T(r_i)) - T(r_{a_D^t}) & \text{if} \ a_D^t \not = r_i \ \text{and} \ m=m^*\\
    \mathbb{E}_p(T(r_i)) - \hat{\mu}_{m,D} & \text{otherwise}.
\end{cases}
\label{Eqn: system Predicted driver Regret}
\end{equation*}
Here, $\hat{\mathbb{E}}_{q^\prime}(T(r_i))$ denotes the system's prediction of the driver's updated expected travel time of recommended route $r_i$, which is calculated similar to $\mathbb{E}_{q^\prime}(T(r_i))$ in Equation (\ref{Eqn: Driver Decision}) by substituting $\alpha$ with the predicted trust value $\hat{\alpha}$. An high-level overview of the interaction between the system and driver at each interaction is shown in Figure \ref{Fig:Interaction Diagram}.

\vspace{-0.05in}
\section{Experimental Evaluation \label{Sect: Experimental Evaluation}}

\begin{table*}[ht]
    \centering
    \caption{Average network congestion and driver travel time (in parentheses) in hours from route recommendation and selection algorithms for Manhattan Pigou network. Lowest values (excluding FC baseline) are in bold.}
    \label{Table: Avg Congestion Manhattan Pigou}
    \vspace{-0.1in}
\begin{tabular}{l|c|ccc|ccc}
\toprule
& FC & ROSTER & TASR & LLF & SR & AR \\
\midrule
0.25 & 0.66561664 (0.277) & \textbf{0.66568944} (\textbf{0.323}) & 0.66568961 (0.324) & 0.66569026 (0.324) & 0.66568947 (0.324) & 0.66573550 (0.389) \\

0.5  & 0.65562692 (0.261) & \textbf{0.65569262} (\textbf{0.301}) & 0.65569456 (0.309) & 0.65569733 (0.311) & 0.65569756 (0.310) & 0.65574811 (0.395) \\

0.75  & 0.64992302 (0.263) & \textbf{0.64997881} (\textbf{0.294}) & 0.64997917 (0.303) & 0.64998696 (0.308) & 0.64998765 (0.305) & 0.65003961 (0.388) \\

1.0  & 0.66240345 (0.266) & \textbf{0.66247416} (\textbf{0.308}) & 0.66247715 (0.326) & 0.66248700 (0.332) & 0.66248498 (0.329) & 0.66253202 (0.397)   \\
\bottomrule
\end{tabular}
\vspace{-0.1in}
\end{table*}

\subsection{Performance Metrics}
For single-stage and multi-stage settings, 
the performance of ROSTER was compared to two route recommendation strategies: Largest-Latency First (LLF) \cite{roughgarden2001stackelberg} and Trust-Aware Stackelberg Routing (TASR) \cite{Brown2024TASR}. LLF prioritizes routing compliant travelers along paths with the highest latency first and allows remaining travelers to choose their routes selfishly in response. 
TASR assumes that the travelers exhibit probabilistic compliance and routes travelers with higher compliance probability along routes with lower latency in response to the noncompliant travelers' flows. A degenerate case of TASR is implemented with one partially compliant agent, which reduces to recommending shortest-time routes and is similar to the recommendation strategy employed by many modern navigation systems. 
The ROSTER, LLF, and TASR algorithms assume the driver's route selection strategy follows the sampling approach (see Section \ref{Sect: Driver BR}). 

Three baseline route selection strategies were considered at the driver: Selfish Routing (labeled \emph{SR}), which assumes that the drivers selfishly choose the route that minimizes travel costs without input from the system; Full Compliance (labeled \emph{FC}), which assumes that the drivers choose the recommended route without considering their prior belief; and Always Reject (labeled \emph{AR}), which assumes drivers are fully noncompliant and always reject the recommendation. SR and FC result in interactions that terminate after the first stage, and FC and AR respectively act as lower upper bounds on the congestion from ROSTER (see Theorem \ref{Lemma: Roster Congestion Bounds}). The performance of each strategy is measured in terms of average total network congestion at the system, average travel time of the driver, execution time, and the ratio of the average congestion to that of FC, referred to as congestion ratio, and the average travel time to that of SR, referred to as travel time ratio. 
The system's prediction of the driver's trust, given by $(\hat{\alpha}_{m^*} - \alpha_{m^*})^2$, is also evaluated to show convergence through repeated, dependent interactions.

\subsection{Simulation Experiments}
Given the significance of Pigou's network \cite{roughgarden2002bad} in routing games, we considered two routing scenarios as follows. In the first scenario, a Pigou network was simulated using traffic data from Manhattan (New York City, USA), with the two routes corresponding to Franklin D. Roosevelt East River (FDR) Drive 
and 2nd Avenue. For each route, speed limits, capacities, and lengths were respectively chosen as 65 and 40 kilometers per hour, 4,000 and 2,000 drivers per hour, and 16 and 8 kilometers. Since the network in the first scenario consists of only two alternative routes, only a single-stage game is considered for each interaction. 

In the second scenario, the Sioux Falls (North Dakota, USA) network was simulated, with edges mapped to modern roadways and capacities standardized to 1,000 and 1,900 vehicles per lane per hour for urban roads and highways, respectively, as in \cite{chakirov2014enriched}. 
In this scenario, the driver was assumed to have an origin of node 10 and destination of node 20, allowing for four alternative routes. In both scenarios, edge volumes were assumed to be randomly uniform values between $0$ (no congestion) and twice the capacity (very congested). The number of available routes in each scenario was chosen to represent common route recommendation applications that typically provide fewer than five routes.

For each experiment, the following gradient scaling constants, learning rates, and discomfort values at the driver and system, respectively, were used: $\varepsilon_D = 0.0002$ and $\varepsilon_S = 0.00015$,
$\eta_{0,D} = 0.0025$ and $\eta_{0,S} = 0.0025$, and $\gamma_D = 1.125$ and $\gamma_S = 1.125$. The number of outcomes to be sampled by the system and the driver were respectively chosen as $g_{S,m} = 5$ and $g_{D,m} = 2$. Four distinct trust values at the driver were considered: (i) $\alpha_0 = 0.25$, (ii) $\alpha_0 = 0.5$, (iii) $\alpha_0 = 0.75$, and (iv) $\alpha_0 = 1.0$. The system was assumed to have a neutral starting driver trust prediction of $\hat{\alpha} = 0.5$, and the system's prior belief $\boldsymbol{p}$ is aligned with the true state of the network $\boldsymbol{\pi}$ in all simulated scenarios. 
Both experiments were implemented in Python 3.12.6. Experimental results are presented in the next section.

\vspace{-0.05in}
\subsection{Performance Results}  \label{Sect: Results}
The average network congestion and driver travel times for the first routing scenario are displayed in Table \ref{Table: Avg Congestion Manhattan Pigou}. For even a single driver in a single-stage interaction, the average network congestion and driver travel time from ROSTER is lower than that of each competing route recommendation algorithm and route selection strategy. The ROSTER algorithm performs better 
than the next best algorithm, TASR, improving further as the driver's trust in the system increases. 
Simulation results detailing the average network congestion ratios and average driver travel time ratios of all strategies for the second scenario are shown in Figures \ref{Fig: Line Plot System Congestion Ratios} and \ref{Fig: Line Plot Driver Travel Time Ratios}, respectively. While ROSTER offers a clear advantage in terms of system congestion in one-, two-, and three-stage interactions, the resulting network congestion from ROSTER is closest to that of FC in three-stage interactions. 
The performance of ROSTER improves as the driver trust increases, indicating that the Monte Carlo-based sampling approach offers an advantage compared to the alternative route recommendation strategies that do not consider the driver's long-term decision-making. While the execution time of ROSTER, shown in Figure \ref{Img: SF Execution Times}, is greater than LLF or TASR, the maximum execution time in the second routing scenario was less than 0.6 milliseconds for a three-stage interaction, which is not computationally expensive.
ROSTER's average congestion values in Table \ref{Table: Avg Congestion Manhattan Pigou} and the average congestion ratios in Figure \ref{Fig: Line Plot System Congestion Ratios} are bounded above and below by the case in which the driver is fully compliant (FC) and fully noncompliant (AR), respectively, providing validation for the theoretical results provided by Theorem \ref{Lemma: Roster Congestion Bounds}. 

Figure \ref{Fig: Trust Prediction Error Comparison} plots the convergence of prediction error in the system's prediction of the driver's trust for the ROSTER algorithm in the second routing scenario under various starting values of $\alpha$ in the interactions up to three stages for $\varepsilon_D = 0.2$ and $\varepsilon_S = 0.15$. 
Although the proposed system trust prediction updating method may initially lead to an increase in the prediction error, it quickly begins to converge to zero in all cases after about forty interactions with the driver.

\begin{figure}
\centering
\vspace{-0.05in}
\includegraphics[width=0.47\textwidth]{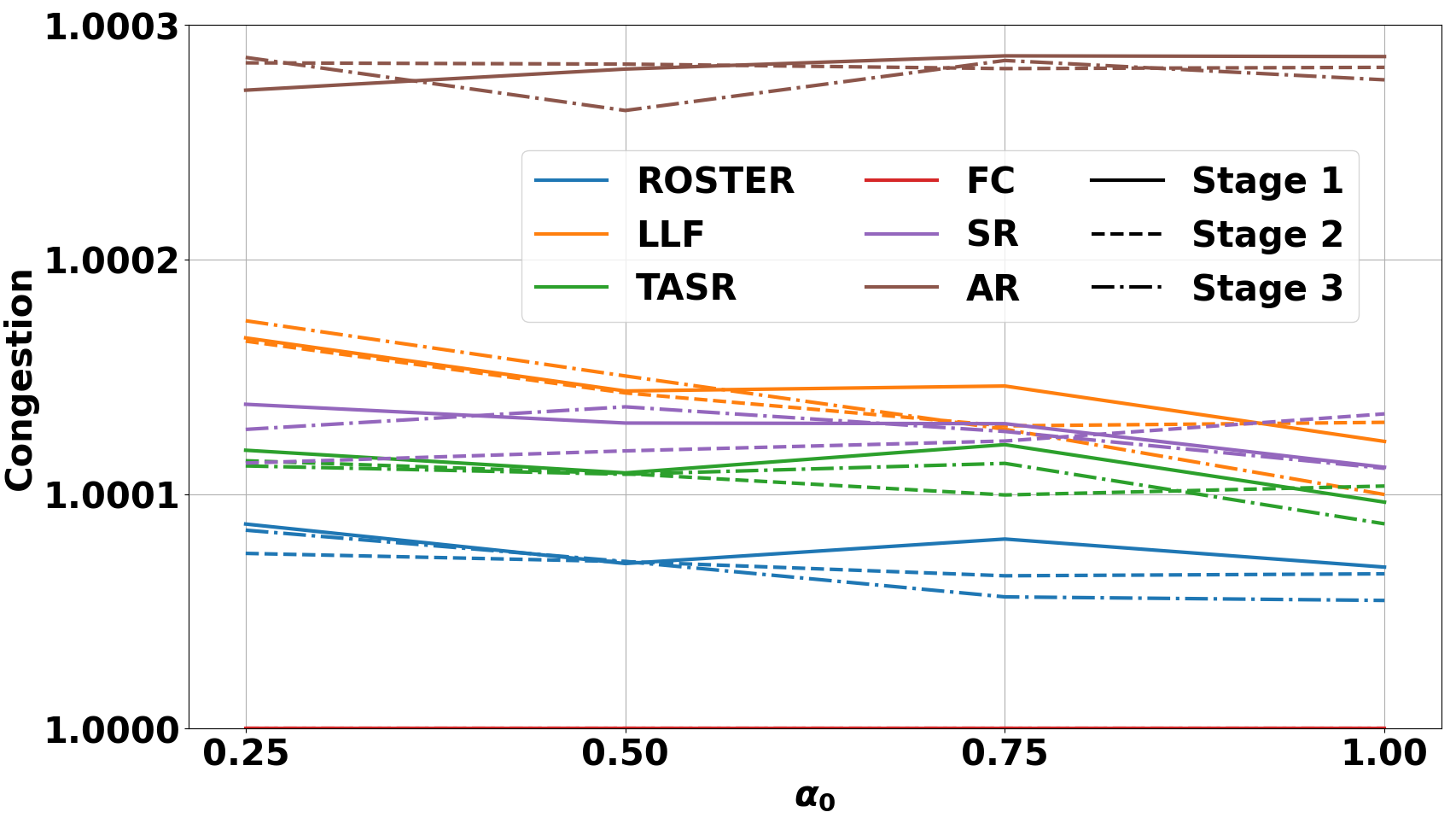}
\vspace{-0.1in}
\caption{Average congestion ratios of all strategies for 1-stage, 2-stage, and 3-stage interactions in Sioux Falls network}
\label{Fig: Line Plot System Congestion Ratios}
\end{figure}

\begin{figure}
\centering
    \vspace{-0.1in}
    \includegraphics[width=0.47\textwidth]{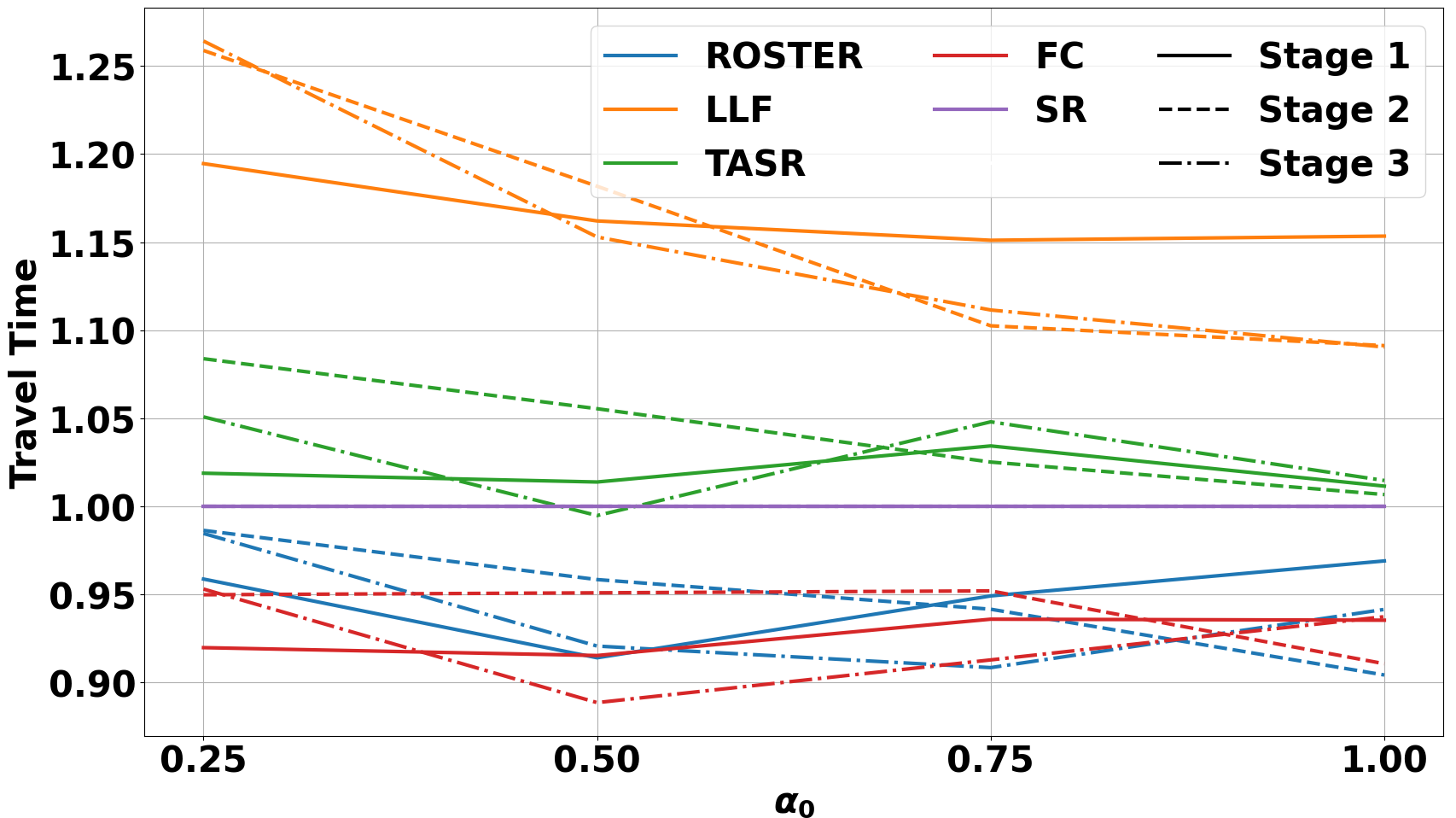}
    \vspace{-0.1in}
\caption{Average travel time ratios of all strategies for 1-stage, 2-stage, and 3-stage interactions in Sioux Falls network}
\label{Fig: Line Plot Driver Travel Time Ratios}
\vspace{-0.2in}
\end{figure}

\begin{figure}
\centering
\vspace{-0.1in}
\includegraphics[width=0.47\textwidth]{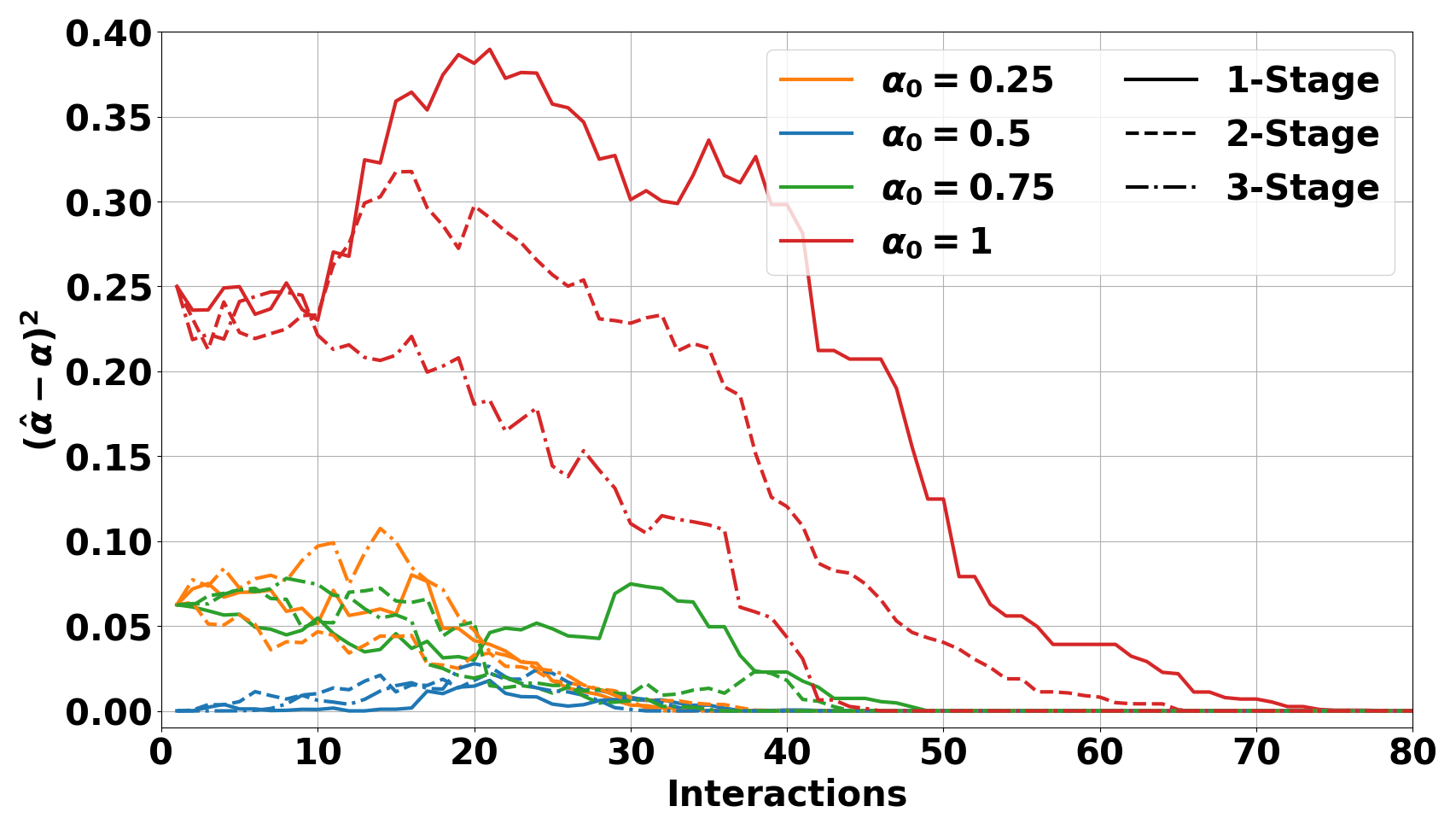}
\vspace{-0.1in}
\caption{Trust prediction error of ROSTER for 1-stage, 2-stage, and 3-stage interactions in Sioux Falls network}
\label{Fig: Trust Prediction Error Comparison}
\end{figure}

\section{Conclusion
\label{Sect: Conclusion}}
This paper presented a novel approximate algorithm, called ROSTER, to compute the Bayesian Nash equilibrium based on Monte Carlo sampling with trust belief updating to determine best response route recommendations sent from the system to the driver. The interaction between the system and driver was modeled as a multi-stage routing game, under the assumption that the system and driver are a network congestion minimizer and travel-time minimizer, respectively. The performance of ROSTER was demonstrated and compared to that of route recommendation strategies TASR and LLF, as well as driver route selection strategies for selfish routing, full compliance, and full noncompliance. Simulation experiments were conducted for traffic networks modeled with Manhattan and Sioux Falls (USA) traffic data. The performance of each strategy was compared in terms of average network congestion, average driver travel time, and convergence of the system's prediction of the driver's trust for interactions with a maximum of three interaction stages. Under the ROSTER algorithm, the system's prediction of the driver's trust was shown to converge to the driver's true trust value across a reasonable number of repeated, dependent interactions. ROSTER was shown to be more effective in mitigating network congestion and driver travel time compared to each alternative strategy, particularly after multiple stages of interaction. 
While this work 
focused on how a single driver’s routing choice affects congestion, future work will explore multi-driver interactions with the system, dynamic traffic flows, and strategic information design to assess when dishonesty benefits congestion mitigation. 

\begin{figure}[!t]
\centering
\vspace{-0.1in}
\includegraphics[width=0.47\textwidth]{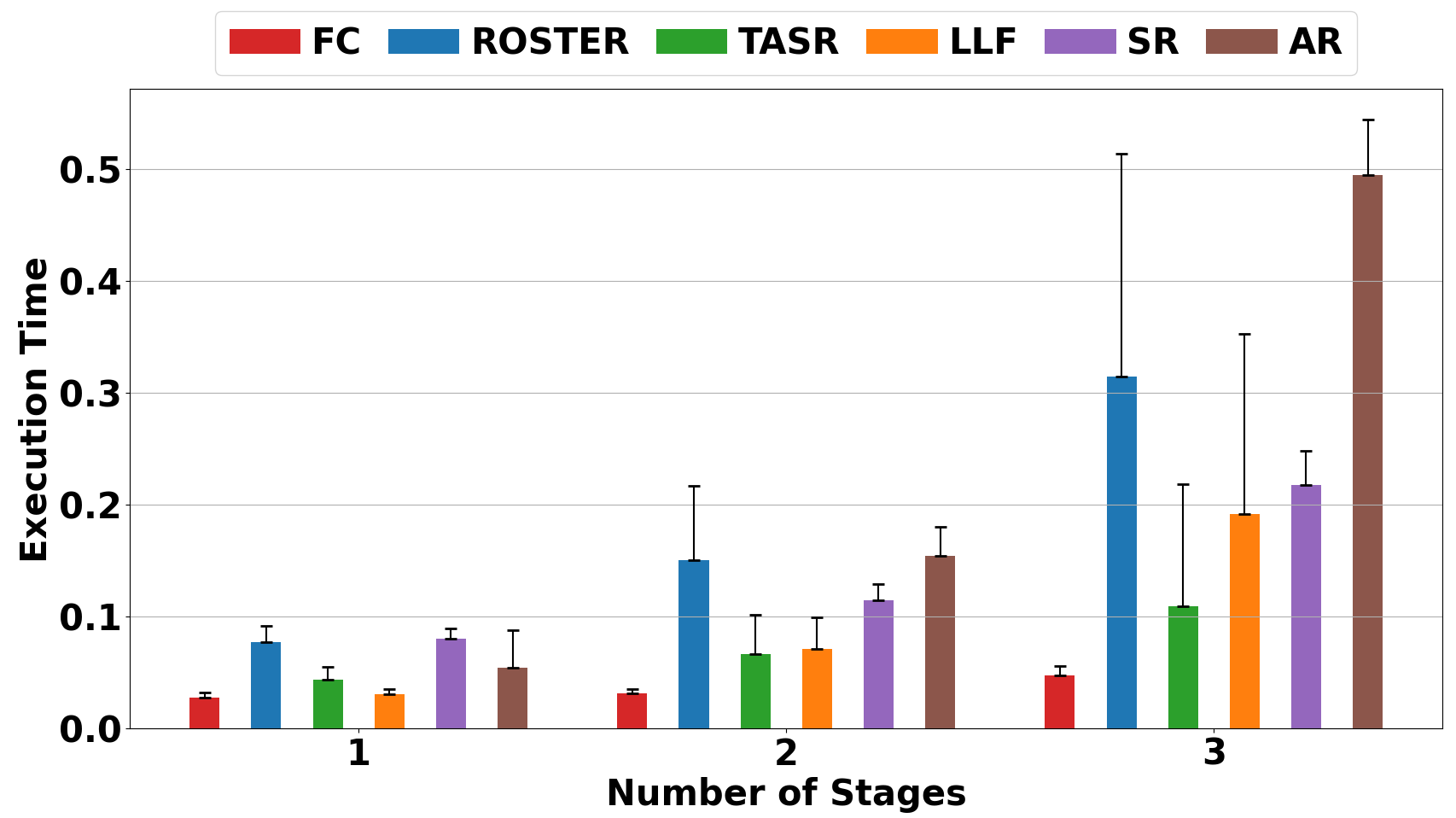}
\vspace{-0.1in}
\caption{Average execution time (milliseconds) and standard deviation for 1-stage, 2-stage, and 3-stage interactions in Sioux Falls network with $\alpha = 0.25$}
\label{Img: SF Execution Times}
\vspace{-0.2in}
\end{figure}

\bibliographystyle{plain}
\bibliography{references-bargaining}

\end{document}